\documentclass[12pt]{amsart}
\usepackage{graphicx,amsmath, amsfonts, amsthm} 
\usepackage{hyperref}

\newtheorem{thm}{Theorem}[section]
\newtheorem{lm}[thm]{Lemma}

\title[The angular momentum distribution in large atoms]{The angular momentum distribution of electrons in large atoms and its deviation from
Madelung's rule}
\author[Bjerg, Fournais, Hearnshaw and Solovej]{August Bjerg,  S{\o}ren Fournais,  Peter Hearnshaw,\\ and Jan Philip Solovej}
\address{Centre for the Mathematics of Quantum Theory\\ University of Copenhagen\\
	Universitetsparken 5\\ DK-2100 Copenhagen \O\\ Denmark}
    
\email{aubjerg23@yahoo.dk}\email{fournais@math.ku.dk}\email{ph@math.ku.dk}\email{solovej@math.ku.dk}

\newcommand{\Phitf}[1]{\Phi^{\rm TF}_{#1}}
\newcommand{\rhotf}[1]{\rho^{\rm TF}_{#1}}
\newcommand{\real}{\mathbb R}
\newcommand{\complex}{\mathbb C}
\newcommand{\Tr}{\textnormal{Tr}}
\begin{document}

\begin{abstract}	 
We show that the angular momentum distribution of neutral atoms, in the sense of occupation numbers, follows the Thomas-Fermi angular momentum distribution in the limit of large atomic number. In particular, we show that the angular momentum distribution of large atoms deviates from the well-known Madelung rule.
\end{abstract}


\maketitle

\section{Introduction}
It is a well established rule in chemistry that the ordering of subshells in neutral atoms is given by the Madelung rule, also known as the Aufbau principle. This principle is very successful for small $Z$, and gives rise to the structure of the periodic table as we know it. The purpose of the current work is to describe the filling up of angular momentum sectors in the non-relativistic quantum setting for neutral atoms in the large $Z$ limit, and, in particular, we describe its deviations from the Madelung rule. The Madelung rule is fairly accurate for the physical atoms in the periodic table, but as we show it eventually fails as atoms grow larger. This raises the interesting question of how large atoms need to be before the Madelung rule no longer applies. 

Our setting is as follows. For each integer $Z \ge 1$ we consider a normalized ground state $\Psi=\Psi_Z$ of the neutral atomic Hamiltonian
\begin{equation*}
    H_Z=\sum_{i=1}^Z\Big(-\Delta_i-\frac{Z}{|x_i|}\Big)+\sum_{1\leq i< j\leq Z} \frac1{|x_i-x_j|}
\end{equation*}
on the space $\bigwedge_{i=1}^Z L^2({\mathbb R}^3,{\mathbb C}^2)$. We use the units $\hbar=2m=e=1$. It is known that a ground state exists for the neutral atom. It may not be unique, but we simply make a choice of $\Psi_Z$  for each $Z$. We also define the one-particle density matrix as the operator $\gamma_\Psi : L^2(\real^3,\complex^2)\to L^2(\real^3,\complex^2)$ given by the kernel
\begin{align*}
    \gamma_\Psi(x,y) = Z\int \Psi_Z(x, x_2, \dots, x_Z)\overline{\Psi_Z(y, x_2, \dots, x_Z)}\,dx_2\dots dx_Z.
\end{align*}
Furthermore, we will denote by $P_\ell$, $\ell=0,1,\ldots$, the projection in $L^2({\mathbb R}^3,{\mathbb C}^2)$ onto the subspace with angular momentum $\ell$.

For $Z > 0$, let $\rho_Z^{TF}$ be the Thomas-Fermi density. Then the Thomas-Fermi potential is given by
\begin{align}
\label{eq:phi_tf}
    \Phi^{TF}_Z(x) = \frac{Z}{|x|} - \int \frac{\rho^{TF}_Z(y)}{|x-y|}\,dy.
\end{align}
Both the density and potential are spherically symmetric and positive. They obey the relation $(\Phi_Z^{TF})^{3/2} = 3\pi^2\rho_Z^{TF}$. Furthermore, the Thomas-Fermi potential obeys the following scaling property,
$$
\Phitf{Z}(x)=Z^{4/3}\Phitf1(Z^{1/3}x).
$$

An important object in this work is the following function, 
\begin{align}
\label{eq:kappa}
    \kappa(\lambda) = \frac1\pi\int_0^\infty(\Phitf1(r)-\lambda^2r^{-2})^{1/2}_+dr
\end{align}
for $\lambda \ge 0$. We here use the convention $(x)_{+} = \max\{0,x\}$ and $(x)_{-} = \min\{0,x\}$. This is a uniformly bounded continuous function with compact support. In addition, it is readily verified that
\begin{align}
    \int_0^\infty 4\lambda\kappa(\lambda) \,d\lambda =1.\label{eq:kappaint}
\end{align}

The function which we denote $\kappa(\lambda)$ was used by E.~Fermi, \cite{fermi28}, to describe the distribution of angular momentum in his own approximate model of the atom, \cite{fermi27}, now known as the Thomas-Fermi model. In particular, he obtains the following formula for the number of electrons of angular momentum $\ell$,
\begin{align*}
    N_\ell = 4\ell\, Z^{1/3} \kappa\Big(\frac{\ell}{Z^{1/3}}\Big).
\end{align*}
The sum over all $\ell$ should give $Z$, and indeed it does in the limit of large $Z$ by a Riemann sum approximation and application of \eqref{eq:kappaint}. Efforts have been made to compare this continuous distribution with the orbital model of quantum theory where electrons appear with discrete angular momentum (for example s, p, d, etc. orbitals). For example, in \cite{jenlutt52} the mean squared angular momentum of the Thomas-Fermi model is compared to known ground-state configurations. T.~A.~Oliphant in \cite{oli56} performed a similar analysis with the Thomas-Fermi-Dirac model. A difficulty in these works is that the angular momentum $\ell$ is not a natural parameter in the density functional theories that they study. A natural way to investigate angular momentum distributions would be to consider the mean-field Hamiltonian
$
-\Delta-\Phitf{Z}
$
and study the angular momentum distribution for the projection onto the $Z$ lowest eigenstates.  
The eigenvalues of the Thomas-Fermi and Thomas-Fermi-Dirac mean-field models were studied in \cite{Latter1955}.

In the present work we study the full non-relativistic quantum description of atoms. Nevertheless the above Thomas-Fermi mean-field model will turn out to play an important role as a controlled approximation. 

When considering electronic configurations, it's natural to ask what the angular momentum is for the latest electron compared to the neutral atom with one less nuclear charge. Indeed, the periodic table is arranged into groups, each of which has an associated angular momentum which is that of the subshell currently being filled. Strictly speaking, in the full quantum setting of interacting electrons, the concept of orbitals of precise angular momentum is ill-defined. Instead, a natural way to quantify the number of electrons of angular momentum $\ell$ in an atom is by the expression 
$$
\langle \Psi|\sum_{i=1}^ZP_{\ell,i}|\Psi\rangle=\Tr[P_\ell\gamma_\Psi],
$$ which can be understood as the occupation number in the sector $\ell$.

We remark that periodicity in the large $Z$ limit has been observed in the Thomas-Fermi mean field model in work by two of the current authors, \cite{bjergsol24}. An interpretation of this result is that the convergence of Thomas-Fermi mean field operators for a sequence of $Z$'s would correspond to descending a group in the periodic table. Although not the subject of the present work, it is believed that these sequences should correspond to atoms with similar angular momentum distributions.

We are now ready to state our main result. We show how $\Tr[P_\ell\gamma_\Psi]$ must grow along certain sequences of $\ell$ which grow proportional to $Z^{1/3}$, thereby tying the Thomas-Fermi angular momentum distribution, $\kappa(\lambda)$, to the quantum angular momentum distribution at large $Z$.


\begin{thm}\label{thm:main} For almost every (in the Lebesgue measure theoretic sense) $\lambda\geq 0$ we have 
$$
\liminf_{\substack{Z\to \infty \\ \ell=\lceil \lambda Z^{1/3}\rceil}} \frac{\text{\rm Tr}[P_\ell\gamma_\Psi]}{2(2\ell+1)Z^{1/3}}
= \kappa(\lambda).
$$
Furthermore, the inequality ``\,$\ge$" holds for every $\lambda>0$.
\end{thm}

The Madelung rule, however, predicts that this limit should be given by the simple expression
$$
\kappa^{\rm Madelung}(\lambda):=[6^{1/3}-2\lambda]_+.
$$
This can be deduced by the following formula, which gives the atomic number $Z$ where we start filling an orbital of angular momentum $\ell$ for the first time, according to the Madelung rule,
\begin{multline*}
    Z_\ell(n) = \frac{1}{6}(n+2\ell - 1)\big((n+2\ell)^2+4(n+2\ell)+9\big)\\ - \frac{1}{4}(1+(-1)^n)(n+2\ell+1)+1 - 2\ell(\ell+2).
\end{multline*}
Here we have introduced the notation
$[x]_{\pm} = \pm \max\{ \pm x, 0\}$ for the positive/negative part of a real number $x$.
As in \eqref{eq:kappaint}, we have the identity
$$
\int4\lambda \kappa^{\rm Madelung}(\lambda)\,d\lambda=1.
$$ 
Therefore, although this can be understood as a valid distribution of angular momentum, Theorem \ref{thm:main} confirms that it is not the correct distribution in the limit of large $Z$.

Finally, we remark that $\kappa^{\rm Madelung}$ can be obtained by replacing $\Phitf{1}$ by the so-called Tietz potential in \eqref{eq:kappa}, see \cite{wong79}, \cite{tietz60}. However, these potentials are distinct, and, therefore, their difference can be understood as what leads to the deviation from Madelung's rule in the limit of large atomic number.

\section{Thomas-Fermi potential }

In this section, we introduce the semiclassical quantities which will be used in the proof of Theorem \ref{thm:main}. First, we consider for $\lambda>0$ and $\mu\geq0$ the semiclassical grand canonical expression
$$e_G(\lambda,\mu)=-\frac2{3\pi}\int_0^\infty(\Phitf1(r)-\lambda^2r^{-2}-\mu)_+^{3/2}dr,$$
and the semiclassical counting function
$$
n(\lambda,\mu)=\frac1\pi\int_0^\infty(\Phitf1(r)-\lambda^2r^{-2}-\mu)^{1/2}_+dr.
$$
We remark that $n(\lambda,0) = \kappa(\lambda)$. It can be seen that the above quantities are finite. Indeed, by \eqref{eq:phi_tf} we have $r\,\Phitf{1}(r) \to 1$ as $r \to 0+$ and, in addition, $\Phitf{1}(r) = O(r^{-4})$ as $r \to \infty$, see for example \cite{Lieb81}. It is readily verified that $\partial_\mu e_G(\lambda,\mu)=n(\lambda,\mu)$.
Since $n(\lambda,\mu)$ is non-negative and non-increasing in $\mu$ we see that $e_G(\lambda, \mu)$ as a function of $\mu$ is non-decreasing and concave. Furthermore, it has finite slope of $n(\lambda, 0)$ at zero. For $\mu$ large enough, depending on $\lambda>0$, we see that $e_G(\lambda, \mu)=n(\lambda,\mu)=0$ since the integrands become zero.

For any $n \ge 0$, we define the canonical energy as 
\begin{align}
\label{eq:ce}
    e(\lambda,n)=\max_{\mu\ge 0}(e_G(\lambda,\mu)-\mu n).
\end{align}
It follows immediately that $e(\lambda, n) \ge e_G(\lambda, 0)$. This becomes equality if and only if $n \ge n(\lambda, 0)$. Indeed, we need only consider the derivative 
\begin{equation}\label{eq:eGmuder}
\partial_\mu(e_G(\lambda, \mu) - \mu n) = n(\lambda, \mu) - n
\end{equation}
in the following cases. If $n \ge n(\lambda, 0)$  then \eqref{eq:eGmuder} is non-positive and therefore $\mu = 0$ maximizes \eqref{eq:ce}. 
Since $n(\lambda,\mu)$ is strictly decreasing as a function of $\mu$ as long as it is positive, we see that for $0 < n < n(\lambda, 0)$ the derivative  \eqref{eq:eGmuder} vanishes at a unique  $\mu>0$. Finally, for $n=0$ we have $e(\lambda,0)=0$.


\section{Many-body formulation}


Unless otherwise specified, we will understand the trace of a one-particle operator to be that over $L^2({\mathbb R}^3,{\mathbb C}^2)$. It is known, for example from 
\cite[equation (5.18) with $\varepsilon=0$]{Lieb81} (see also \cite{HUGHES1990213} or \cite[proof of Lemma 2]{siewei89}), that 
\begin{align}\label{eq:HZlower}
    \langle\Psi_Z,H_Z\Psi_Z\rangle\geq & \Tr\big[(-\Delta - \Phitf{Z})\gamma_\Psi\big]\nonumber\\&
    -\frac12\iint \frac{\rhotf{Z}(x)\rhotf{Z}(y)}{|x-y|}dxdy-CZ^{5/3}
\end{align}
for some $C>0$.
In addition, it follows, for example from \cite{Seidentopweikard87} or \cite{solspitz03}, that for some $\varepsilon>0$ we have 
\begin{align}
    \langle\Psi_Z, H_Z\Psi_Z\rangle\leq &\nonumber
    \text{Tr}\big[(-\Delta-\Phitf{Z})_{-} \big]\\& -\frac12\iint \frac{\rhotf{Z}(x)\rhotf{Z}(y)}{|x-y|}dxdy+CZ^{2-\varepsilon}.\label{eq:HZupper}
\end{align}
We then conclude from \eqref{eq:HZlower} and \eqref{eq:HZupper} that 
\begin{equation*}  \sum_{\ell=0}^\infty\Big(\text{Tr}\left[P_\ell\gamma_\Psi P_\ell (-\Delta-\Phitf{Z})\right]-\text{Tr}\left[P_\ell(-\Delta-\Phitf{Z} )_-\right]\Big)\leq
    CZ^{2-\varepsilon}.
\end{equation*}
Since every term in the sum is non-negative we find that for every $\ell$ and every  $\mu\geq0$ 
\begin{align}
    \text{Tr}\left[P_\ell(-\Delta-\Phitf{Z} )_-\right]\geq & \text{Tr}\left[P_\ell\gamma_\Psi P_\ell (-\Delta-\Phitf{Z})\right]
    -CZ^{2-\varepsilon}\nonumber\\=&
    \text{Tr}\left[P_\ell\gamma_\Psi P_\ell (-\Delta-\Phitf{Z}+\mu Z^{4/3})\right]\nonumber\\&-\mu Z^{4/3}\text{Tr}[P_\ell\gamma_\Psi]
    -CZ^{2-\varepsilon}.\label{eq:llower}
\end{align}
This bound holds for all $\ell$, but we choose a sequence of $\ell$'s which grow proportional to $Z^{1/3}$. Namely, given some $\lambda>0$, for each $Z$ let $\ell_Z$ be the unique positive integer such that
\begin{equation}\label{eq:ellbound}
    (\ell_Z-1)Z^{-1/3} <\lambda\leq\ell_Z Z^{-1/3}.
\end{equation}

\begin{lm}
\label{lem:lb}
Let $\lambda>0$ and $\mu \ge 0$. Then
\begin{align*}
\Tr\big[P_{\ell_Z}\gamma_\Psi P_{\ell_Z} (-\Delta-\Phitf{Z}+\mu Z^{4/3})\big] \geq
    4\lambda Z^2 e_G(\lambda,\mu) + o_{\lambda,\mu}(Z^2)
\end{align*}
as $Z \to \infty$. Here the subscripts $\lambda,\mu$ on the error term indicate that it may depend on these parameters.
\end{lm}
\begin{proof}
By passing to spherical coordinates, we can bound the left-hand side from below by
\begin{align*}
2(2\ell_Z+1)&\Tr_{L^2({\mathbb R}_+)}\big[(h_{\ell_Z,Z}+\mu Z^{4/3})_-\big]
\end{align*}
where 
$$
h_{\ell,Z}=-\frac{d^2}{dr^2}-\Phitf{Z}(r)+\frac{\ell(\ell+1)}{r^2}.
$$
The operator $h_{\ell,Z}$ is unitarily equivalent to $Z^{4/3}\widetilde{h}_{\ell,Z}$ with 
$$
\widetilde{h}_{\ell,Z}=-Z^{-2/3}\frac{d^2}{dr^2}-\Phitf{1}(r)+\frac{\ell(\ell+1)}{Z^{2/3}r^2}.$$
By \eqref{eq:ellbound}, we then have
$$
\widetilde{h}_{\ell_Z,Z}\geq -Z^{-2/3}\frac{d^2}{dr^2}-\Phitf{1}(r)+\frac{\lambda^2}{r^2}.
$$
A standard leading order semiclassical estimate, along with \eqref{eq:ellbound} again, therefore gives
\begin{align*}2(2\ell_Z+1)&\text{Tr}_{L^2({\mathbb R}_+)}\big([h_{\ell_Z,Z}+\mu Z^{4/3}]_-\big) \\
&\geq 4(2\pi)^{-1}Z^2\lambda\int\big[p^2-\Phitf1+\lambda^2r^{-2}+\mu\big]_-drdp-o_{\lambda,\mu}(Z^2)\nonumber\\
&= 4Z^2\lambda e_G(\lambda,\mu)-o_{\lambda,\mu}(Z^2)
\end{align*}
as $Z\to\infty$.
\end{proof}

\begin{lm}
\label{lem:ub}
Let $\lambda, \epsilon>0$, and $\ell_Z$ be given by \eqref{eq:ellbound}. Then
    \begin{equation}\label{eq:enupper}
    \Tr\big[P_{\ell_Z}(-\Delta-\Phitf{Z})_-\big]
    \leq 4\lambda Z^2 e_G(\lambda+\varepsilon,0)+ o_{\lambda, \varepsilon}(Z^2)
\end{equation}
as $Z \to \infty$.
\end{lm}
\begin{proof}
For all $\varepsilon>0$ we have $\lambda+\varepsilon\geq (\ell_Z+1)Z^{-1/3}$ for all $Z\geq (\varepsilon/2)^{-3}$. We then continue as in the proof of Lemma~\ref{lem:lb}, but with a semiclassical upper bound.
\end{proof}

For $\lambda>0$, let $\ell_Z$ be the sequence defined by (\ref{eq:ellbound}). Then define $\chi_{\Psi}(\lambda, \,\cdot\,)$ by
$$
\chi_{\Psi}(\lambda, Z)=\frac{\text{Tr}[P_{\ell_Z}\gamma_{\Psi}]}{4\lambda Z^{2/3}}
$$ 
for $Z=1,2,\dots$. Now fix $\varepsilon>0$ and $\mu\geq0$. Then by Lemmas~\ref{lem:lb} and~\ref{lem:ub} applied to \eqref{eq:llower} we get
\begin{equation}\label{eq:nellbound}
    e_G(\lambda,\mu)-\mu \chi_{\Psi}(\lambda, Z)\leq e_G(\lambda+\varepsilon,0)+ o_{\lambda, \mu, \varepsilon}(1)
\end{equation}
for integer $Z \to \infty$.


\begin{lm} \label{lm:lower} For all $\lambda>0$ we have 
    $$
        \liminf_{Z\to\infty}\chi_\Psi(\lambda,Z)\geq    \kappa(\lambda).
    $$  
\end{lm}
\begin{proof}
    We recall that $\kappa(\lambda)=n(\lambda,0)$. Assume for contradiction that for some $\lambda>0$ we have 
    $$
    \liminf_{Z\to\infty}\chi_\Psi(\lambda,Z)< n(\lambda,0).
    $$ We may then find some $n<n(\lambda,0)$ and a sequence $Z_k\to\infty$ such that 
    $\chi_\Psi(\lambda,Z_k)\leq n$ for all $k$. Thus by \eqref{eq:nellbound}, we see that for all $\mu\geq0$,
    $$
    e_G(\lambda,\mu)-\mu n \leq e_G(\lambda+\varepsilon,0),
    $$
    where we took the limit $k \to \infty$. Letting $\varepsilon\to0$, we arrive at 
    $$
    e(\lambda,n)\leq e_G(\lambda,0).
    $$
    Here, we used that $e_G(\lambda,0)$ is continuous in $\lambda$ on $(0,\infty)$. By the remark after \eqref{eq:ce}, this is in fact an equality and holds only if $n \ge n(\lambda, 0)$. This contradicts our assumption that $n < n(\lambda, 0)$.
\end{proof}
To arrive at an inequality in the opposite direction we observe that for all $Z$
\begin{equation*}
\int_0^\infty 4\lambda\chi_\Psi(\lambda,Z)\,d\lambda=
Z^{-1} \sum_{\ell=1}^\infty \text{Tr}[P_\ell\gamma_\Psi] \leq Z^{-1}\sum_{\ell=0}^\infty \text{Tr}[P_\ell\gamma_\Psi]=1.
\end{equation*}
Therefore, by Fatou's lemma, Lemma~\ref{lm:lower} and \eqref{eq:kappaint} we conclude that 
\begin{multline*}
  1\geq \liminf_{Z\to\infty} \int_0^\infty 4\lambda\chi_\Psi(\lambda,Z)\,d\lambda\\\geq  
  \int_0^\infty 4\lambda\liminf_{Z\to\infty} \chi_\Psi(\lambda,Z)\,d\lambda\geq 
  \int_0^\infty 4\lambda\kappa(\lambda)\,d\lambda=1.
\end{multline*}
It then follows from Lemma~\ref{lm:lower} again that 
$$
\liminf_{Z\to\infty} \chi_\Psi(\lambda,Z)=\kappa(\lambda)
$$
for almost all $\lambda\geq 0$. A straightforward restatement completes the proof of Theorem \ref{thm:main}.

\vskip 0.5cm
\noindent
\textbf{Funding.} This work was supported by the Villum Centre of Excellence
for the Mathematics of Quantum Theory (QMATH) with Grant No.10059.
A.B. and S.F. were also partially supported by the European Union (via the ERC Advanced Grant MathBEC - 101095820). Views and opinions expressed are however those of the authors only and do not necessarily reflect those of the European Union or the European Research Council. Neither the European Union nor the granting authority can be held responsible for them.

\bibliographystyle{unsrt}
\bibliography{references}

\end{document}